\documentclass[letterpaper, 10 pt, conference]{ieeeconf}
\IEEEoverridecommandlockouts
\usepackage{graphicx}
\usepackage{amsmath}
\usepackage{amssymb}

\usepackage{amsthm}
\usepackage{cite}
\usepackage{hyperref}
\usepackage{accents}
\usepackage{centernot}
\usepackage{color}
\usepackage{algorithm}
\usepackage[noend]{algpseudocode}
\usepackage[font=scriptsize]{caption}
\usepackage{subcaption} 
\usepackage{mathtools}
\usepackage{soul}

\renewcommand{\baselinestretch}{0.94}

\newtheorem{ass}{Assumption}
\newtheorem{thm}{Theorem}

\newtheorem{prop}{Proposition}

\definecolor{mygreen}{rgb}{0,0.6,0.3}
\definecolor{mypurple}{rgb}{0.75,0.0,0.75}

\title{\LARGE \bf
A Simple and Efficient Tube-based Robust Output Feedback Model Predictive Control Scheme
}

\author{Joseph Lorenzetti, Marco Pavone	
\thanks{The authors are with the Department of Aeronautics and Astronautics, Stanford University, Stanford CA. Emails: \{jlorenze, pavone\}@stanford.edu.}
\thanks{This work was supported by the Office of Naval Research  (Grant N00014-17-1-2749). Joseph Lorenzetti is supported by the Department of Defense (DoD) through the National Defense Science and Engineering Fellowship (NDSEG) Program.}
}

\begin{document}

\maketitle
\thispagestyle{empty}
\pagestyle{empty}

\begin{abstract}
The control of constrained systems using model predictive control (MPC) becomes more challenging when full state information is not available and when the nominal system model and measurements are corrupted by noise. Since these conditions are often seen in practical scenarios, techniques such as robust output feedback MPC have been developed to address them. However, existing approaches to robust output feedback MPC are challenged by increased complexity of the online optimization problem, increased computational requirements for controller synthesis, or both. In this work we present a simple and efficient methodology for synthesizing a tube-based robust output feedback MPC scheme for linear, discrete, time-invariant systems subject to bounded, additive disturbances. Specifically, we first formulate a scheme where the online MPC optimization problem has the same complexity as in the nominal full state feedback MPC by using a single tube with constant cross-section. This makes our proposed approach simpler to implement and less computationally demanding than previous methods for both online implementation and offline controller synthesis. Secondly, we propose a novel and simple procedure for the computation of robust positively invariant (RPI) sets that are approximations of the minimal RPI set, which can be used to define the tube in the proposed control scheme.
\end{abstract}

\section{Introduction}
Model predictive control (MPC) is a useful framework for the optimal control of constrained systems due to its ability to \textit{explicitly} account for state and control constraints. This is accomplished by exploiting a model of the system, where the control law can be implicitly defined as the solution to a finite-horizon optimization problem that is solved online in a receding horizon fashion. For a broad survey of MPC theory and additional references see \cite{RawlingsMayneEtAl2017}. Early work in the development of MPC theory focused on the case where the full state was assumed to be known, and where no disturbances affected the system behavior, which we will refer to as \textit{nominal} MPC. However in practice these assumptions typically fail, and so \textit{robust output feedback} MPC schemes were developed to address the more general case where state estimators are employed and when the nominal system is subject to bounded, additive disturbances.

{\em Related Work:}
One approach used to handle the robust output feedback MPC problem uses a min-max optimization formulation \cite{FindeisenAllgoewer2004, CoppHespanha2014}. However such formulations result in optimization problems with increased complexity and can therefore be less desirable for real-time control applications. On the other hand, tube-based approaches generally formulate the optimization problem based on nominal system dynamics, and then incorporate an ancillary feedback controller to ensure the nominal system is tracked with bounded error. Such formulations rely on an offline analysis to verify robustness, which enables the simple form of the online optimization problem that is advantageous for real-time control. Of course there are disadvantages with these approaches as well, namely that they can be sub-optimal \cite{Mayne2016}, can be conservative, and can be computationally difficult to synthesize.

One early example of a tube-based scheme is \cite{MayneRakovicEtAl2006} which relies on the computation of robust positively invariant (RPI) sets to bound the error between the nominal and real systems, which are then used to tighten the constraints appropriately. They also modify the optimization problem to make the initial nominal state a decision variable. The work in \cite{SuiFengEtAl2008} takes a similar approach, but uses a moving horizon estimator and tightens the constraints sequentially. To reduce conservatism in the constraint tightening seen in \cite{MayneRakovicEtAl2006}, \cite{KoegelFindeisen2017} considers the \textit{coupled} error dynamics and also tightens the constraints sequentially. The method defined by \cite{BrunnerMuellerEtAl2016} also tries to reduce conservatism by using a set-valued moving horizon estimator that seeks to provide tighter error bounds. Inspired by \cite{MayneRakovicEtAl2006} and \cite{KoegelFindeisen2017}, our first contribution is to define a tube-based method that is more amenable to real-world applications where computational efficiency is critical by simplifying the controller synthesis and decreasing the online computational complexity with respect to previous approaches. In fact, the online computational complexity of the approach matches the nominal full feedback MPC case, in contrast to \cite{MayneRakovicEtAl2006, SuiFengEtAl2008, BrunnerMuellerEtAl2016}. This formulation could be seen as an extension of the ideas presented in \cite{MayneLangson2001} to the output feedback case, and is also similar to \cite{KoegelFindeisen2017}, which originally showed the advantages of using a single tube over the approach in \cite{MayneRakovicEtAl2006} (which does not consider coupled error dynamics). However in contrast to \cite{KoegelFindeisen2017} we propose to use a constant cross-section tube based on an approximation to the minimal RPI set, which makes controller synthesis and online implementation more simple and efficient.

As our proposed approach utilizes RPI sets, our second contribution is a novel RPI set computation method that is simple to implement and is computationally efficient. This approach  leverages the work in \cite{SchulzeDarupTeichrib2019} where an approach for the efficient computation of RPI sets is developed based on a clever combination of the methods in \cite{KolmanovskyGilbert1998} and \cite{RakovicKerriganEtAl2005}. We also take advantage of the work by \cite{Trodden2016}, who proposes a method for computing RPI sets that requires only a single linear program. 

{\em Statement of Contributions:}
To summarize, in this paper we present a simple and efficient tube-based robust output feedback MPC scheme. The proposed method is efficient in both the \textit{offline} synthesis of the controller and in the \textit{online} implementation. Such efficiency is crucial for enabling robust constrained control of real-world systems where the system's state dimension may be large or when the dynamics evolve quickly.
Specifically, in our approach we first propose to use a formulation of the online optimization problem that has reduced complexity over previous methods, and where the overall control scheme is efficient to synthesize. Second, we propose a novel, simple, and computationally efficient technique for computing RPI sets. While we demonstrate the use of this RPI computation method with respect to the proposed MPC scheme, it is a general methodology whose scope is not limited to this work. The effectiveness of the proposed scheme is demonstrated in simulation using two examples: the control of a simple synthetic system and the control of a wind energy conversion system.

{\em Organization:} We begin our discussion in Section \ref{sec:prob} with a formal description of the problem that we are trying to solve, and define several useful mathematical concepts in Section \ref{sec:notation}. Next, in Section \ref{sec:proposed} we describe our proposed robust output feedback MPC scheme and in Section \ref{sec:computingrpi} we present a new method for computing RPI sets that can be used to synthesize our proposed controller. Then our approach is demonstrated in Section \ref{sec:examples} on a simple synthetic example as well as on an more practical example of a wind energy conversion system. Finally, we conclude with some observations and remarks in Section \ref{sec:conclusion}.

\section{Problem Formulation} \label{sec:prob}
In this work we consider systems described by linear, discrete, time-invariant state space models of the form
\begin{equation} \label{eq:dyn}
\begin{split}
x_{k+1} &= Ax_k + Bu_k + w_k, \\
y_k &= C x_k + v_k, \quad z_k = H x_k,\\
\end{split}
\end{equation}
where $x \in \mathbb{R}^n$ is the state of the system, $u \in \mathbb{R}^m$ is the control input, $y \in \mathbb{R}^p$ is the measured output, $z \in \mathbb{R}^o$ are performance variables, $w \in \mathbb{R}^n$ are unknown process noise terms, $v \in \mathbb{R}^p$ are unknown measurement noise terms, and $A$, $B$, $C$, and $H$ are matrices of appropriate dimension.

Set-based constraints on the performance variables $z$ and the controls $u$ are also considered, which are defined by
\begin{equation} \label{eq:constraints}
z_k \in \mathcal{Z}, \quad u_k \in \mathcal{U},
\end{equation}
where $\mathcal{Z} \coloneqq \{z \:|\: H_z z \leq b_z \}$ and $\mathcal{U} \coloneqq \{u \:|\: H_u u \leq b_u \}$ are convex polyhedra and the inequalities are interpreted element-wise. It is also assumed that the noise terms $w$ and $v$ are constrained such that
\begin{equation} \label{eq:disturbances}
w_k \in \mathcal{W}, \quad v_k \in \mathcal{V},
\end{equation}
where $\mathcal{W} \coloneqq \{w \:|\: H_w w \leq b_w \}$ and $\mathcal{V} \coloneqq \{v \:|\: H_v v \leq b_v \}$ are also convex polyhedra. The following assumptions are also made about the system, the constraints, and the disturbances:
\begin{ass} \label{ass:ctrbobsv}
The pair $(A,B)$ is controllable and the pair $(A,C)$ is observable.
\end{ass}
\begin{ass} \label{ass:convexcompact}
The sets $\mathcal{Z}$, $\mathcal{U}$, $\mathcal{W}$, and $\mathcal{V}$,  are compact and contain the origin in their interior. 
\end{ass}

The control problem of interest is to optimally regulate the system \eqref{eq:dyn} to the origin while ensuring the constraints \eqref{eq:constraints} are robustly satisfied. The optimality of the control is assumed to be defined with respect to a quadratic, infinite-horizon cost function
\begin{equation} \label{eq:costfunc}
J = \sum_{k = 0}^\infty x_k^T Q x_k + u_k^T R u_k,
\end{equation}
where $Q \in\mathbb{R}^{n\times n}$ and $R \in\mathbb{R}^{m\times m}$ are symmetric, positive definite weighting matrices. Since in practice the cost is typically defined with respect to the performance variables $z$, using the positive definite matrix $Q_z$, the matrix $Q$ could be defined as $Q = H^TQ_zH + \gamma I$ where $\gamma \geq 0$ is chosen to ensure positive-definiteness of $Q$.

\subsection{State Estimator}
Since it is assumed that knowledge about the state $x$ is not directly available, a state estimator is required. For this work we assume a Luenberger estimator is used and is defined by
\begin{equation} \label{eq:estimator}
\hat{x}_{k+1} = A\hat{x}_k + Bu_k + L(y_k - \hat{y}_k), \quad \hat{y}_k = C\hat{x}_k,
\end{equation}
where $\hat{x} \in \mathbb{R}^n$ is the state estimate and $L$ is the observer gain matrix of appropriate dimension. It is assumed that $L$ is chosen such that the matrix $A-LC$ is Schur stable.

\section{Mathematical Preliminaries} \label{sec:notation}
Before describing our proposed control methodology it is useful to define some terminology that will be used throughout the remainder of this work. We begin with the definition of a robust positively invariant (RPI) set.

\subsection{Robust Positively Invariant Sets}
Consider an autonomous system with dynamics
\begin{equation} \label{eq:rpidynamics}
e_{k+1} = A_e e_k + \delta_k, \quad \phi_k = Ee_k
\end{equation}
where $e \in \mathbb{R}^{n_e}$ is the state, $\delta \in \mathbb{R}^{n_e}$ is a disturbance, $A_e \in \mathbb{R}^{n_e \times n_e}$ defines the system dynamics, and where $\phi \in \mathbb{R}^{o_e}$ defines a general output variable that has set-based constraints given by $\phi \in \Phi$. Additionally, the disturbance $\delta$ is constrained to lie in the set $\Delta \coloneqq \{\delta \:|\: H_\delta \delta \leq b_\delta \}$ and the constraint set for $\phi$ is defined as $\Phi \coloneqq \{\phi \:|\: H_\phi \phi \leq b_\phi \}$. We also assume that the sets $\Delta$ and $\Phi$ are convex, compact, contain the origin in their interiors, and that $A_e$ is Schur stable.

A set $\mathcal{R}$ is an RPI set for this system if for all $e_k \in \mathcal{R}$ and for all $\delta_k \in \Delta$ the state $e_{k+1}$ also satisfies $e_{k+1} \in \mathcal{R}$. In shorthand we write this condition as $A_e\mathcal{R} \oplus \Delta \subseteq \mathcal{R}$, where $\oplus$
represents the Minkowski sum, defined for two sets $\mathcal{A}$ and $\mathcal{B}$ as $\mathcal{A} \oplus \mathcal{B} \coloneqq \{a + b \:|\: a \in \mathcal{A}, \: b \in \mathcal{B} \}$.

Under the stated assumptions there is guaranteed to exist an RPI set for the system, and in general there may be many. Of particular interest is the \textit{minimal} RPI set, denoted by $\mathcal{R}_\infty$. The minimal RPI set has the special property that it is contained within every RPI set for \eqref{eq:rpidynamics}. Additionally, \textit{constraint admissible} RPI sets are those which also satisfy the condition $E\mathcal{R} \subseteq \Phi$.

\subsection{Set Computations}
Consider two convex, compact sets, $\mathcal{A} \subset \mathbb{R}^{n_A}$ and $\mathcal{B} \subset \mathbb{R}^{n_B}$ and a linear map $C \in \mathbb{R}^{n_A \times n_B}$. The set $C\mathcal{B}$ is defined as $C\mathcal{B} \coloneqq \{Cb \in \mathbb{R}^{n_A} \:|\: b \in \mathcal{B} \}$. The Pontryagin difference $\mathcal{A} \ominus C \mathcal{B}$ is defined as $\mathcal{A} \ominus C\mathcal{B} \coloneqq \{d \in \mathbb{R}^{n_A} \:|\: d + Cb \in \mathcal{A}, \forall b \in \mathcal{B} \}$. 

Suppose that in addition to being convex, $\mathcal{A}$ and $\mathcal{B}$ are defined using a half-space representation given by $\mathcal{A} \coloneqq \{a \:|\:H_a a \leq b_a \}$ and $\mathcal{B} \coloneqq \{b \:|\:H_b b \leq b_b \}$, then the Pontryagin difference can be computed as:
\begin{equation*}
\mathcal{A} \ominus C\mathcal{B} \coloneqq \{a \in \mathbb{R}^{n_A} \:|\:H_a a \leq b_a - \delta_a \},
\end{equation*}
where each element of the vector $\delta_a$ is given by $\delta_{a,i} \coloneqq S_\mathcal{B}( h_{a,i}^TC)$, where $h_{a,i}^T$ is the $i^\text{th}$ row of the matrix $H_a$ and  $S_\mathcal{B}(h_{a,i}^TC)$ is defined as the linear program $S_\mathcal{B}(h_{a,i}^TC) \coloneqq \max_{b \in \mathcal{B}} h_{a,i}^TCb$. We also overload this notation to simply write $\delta_a \coloneqq S_\mathcal{B}( H_aC)$.

\section{Robust MPC Scheme} \label{sec:proposed}
Now that we have defined the problem in Section \ref{sec:prob} and introduced some mathematical notation in Section \ref{sec:notation}, we move on to discussing our proposed control scheme. The scheme consists of two parts: first, a receding horizon optimization problem is used for computing an optimal nominal trajectory over a finite horizon $N$. Then a control law is defined that drives the real system to track the nominal trajectory. While the proposed scheme is an extension of \cite{MayneLangson2001} and leverages textbook MPC results, we choose to provide a detailed discussion for the sake of clarity and completeness. We begin by defining the nominal system dynamics.

\subsection{Nominal System} \label{subsec:nomsys}
Since the noise terms $w$ and $v$ are unknown disturbances, we define a \textit{nominal}, noise-free, system that will be used for planning. This system is given by
\begin{equation} \label{eq:nomdyn}
\bar{x}_{k+1} = A\bar{x}_k + B\bar{u}_k, \quad \bar{z}_k = H \bar{x}_k,
\end{equation}
where $\bar{x}\in\mathbb{R}^n$ denotes the nominal system state, $\bar{u} \in \mathbb{R}^m$ is the nominal system control, and $\bar{z} \in \mathbb{R}^o$ are the nominal performance variables. This system is initialized at time $k = 0$ by the current state estimate such that $\bar{x}_0 = \hat{x}_0$.

\subsection{Control Law}
The control law that is applied to the real system \eqref{eq:dyn} is then defined by
\begin{equation} \label{eq:controller}
u_k = \bar{u}_k + K(\hat{x}_k - \bar{x}_k),
\end{equation}
where the terms $\bar{x}_k$ and $\bar{u}_k$ are defined by the nominal system trajectory and the current state estimate $\hat{x}_k$ is given by \eqref{eq:estimator}. The feedback gain matrix $K$ is assumed to be chosen such that the matrix $A+BK$ is Schur stable.

\subsection{Online Optimization Problem}
As mentioned earlier, the online finite-horizon optimization problem is based on the nominal system dynamics \eqref{eq:nomdyn}. Since this model is artificial and disturbance free, full state knowledge is available and therefore a simple and efficient MPC problem can be used. Specifically, we choose to formulate the problem as
\begin{equation} \label{eq:robustmpc}
\begin{split}
(\mathbf{\bar{x}^*_k}, \mathbf{\bar{u}^*_k}) = \underset{\mathbf{\bar{x}_k}, \mathbf{\bar{u}_k}}{\text{argmin.}} \:\:& \lVert \bar{x}_{k+N|k}\rVert^2_P + \sum_{j=k}^{k+N-1}\lVert \bar{x}_{j|k}\rVert^2_Q + \lVert \bar{u}_{j|k}\rVert^2_R, \\
\text{subject to} \:\: & \bar{x}_{i+1|k} = A\bar{x}_{i|k} + B\bar{u}_{i|k},  \\
& H\bar{x}_{i|k} \in \bar{\mathcal{Z}}, \:\: \bar{u}_{i|k} \in \bar{\mathcal{U}}, \\
& \bar{x}_{k+N|k} \in \mathcal{X}_N, \quad \bar{x}_{k|k} = \bar{x}_k, \\
\end{split}
\end{equation}
where $i = k, \dots, k+N-1$, the integer $N$ defines the planning horizon, $\bar{x}_k$ is the current nominal system state, and the solution yields the optimal nominal trajectory: $\mathbf{\bar{x}^*_k} \coloneqq [\bar{x}^*_{k|k},\dots,\bar{x}^*_{k+N|k}]$ and $\mathbf{\bar{u}^*_k} \coloneqq [\bar{u}^*_{k|k},\dots,\bar{u}^*_{k+N-1|k}]$. The resulting nominal system control at time $k$ is then defined as $\bar{u}_k = \bar{u}^*_{k|k}$, and the nominal system state at time $k+1$ is then given by \eqref{eq:nomdyn}.

In \eqref{eq:robustmpc} the symmetric, positive definite cost matrices $Q$ and $R$ are chosen to be the same as in \eqref{eq:costfunc}. Additional design variables for the problem include the terminal cost matrix $P$, the terminal set $\mathcal{X}_N$, and the constraint sets $\bar{\mathcal{Z}}$ and $\bar{\mathcal{U}}$. First, in Section \ref{subsec:stability} a procedure is outlined for computing the terminal cost $P$ and terminal set $\mathcal{X}_N$ that will guarantee stability of the nominal system. Then, in Section \ref{subsec:tightconst} we describe how the sets $\bar{\mathcal{Z}}$ and $\bar{\mathcal{U}}$ are defined to ensure robust constraint satisfaction for the real system under the proposed control scheme.

\subsection{Nominal System Stability} \label{subsec:stability}
By the appropriate design of the terminal cost matrix $P$ and terminal set $\mathcal{X}_N$, the online optimization problem \eqref{eq:robustmpc} can ensure closed-loop stability for the nominal system \eqref{eq:nomdyn}. Specifically we choose to use a well-known approach described in \cite{RawlingsMayneEtAl2017} which requires finding a terminal controller $\kappa(\bar{x})$, a terminal cost matrix $P$, and a terminal set $\mathcal{X}_N$ which satisfy the following properties:
\begin{equation} \label{eq:cond1}
\begin{split}
A\bar{x} + B\kappa(\bar{x}) \subseteq \mathcal{X}_N, \quad &\forall \bar{x} \in \mathcal{X}_N, \\
H\mathcal{X}_N \subseteq \bar{\mathcal{Z}}, \quad \kappa(\bar{x}) \subseteq \bar{\mathcal{U}}, \quad &\forall \bar{x} \in \mathcal{X}_N,
\end{split}
\end{equation}
\begin{equation} \label{eq:cond2}
V_N(\bar{x}_{k+1}) + l(\bar{x}_k, \kappa(\bar{x}_k)) \leq V_N(\bar{x}_{k}), \quad \forall \bar{x} \in \mathcal{X}_N,
\end{equation}
where $V_N(\bar{x}) \coloneqq \lVert \bar{x} \rVert^2_P$ and $l(\bar{x}, \bar{u}) \coloneqq \lVert \bar{x} \rVert^2_Q + \lVert \bar{u} \rVert^2_R$.

The first condition \eqref{eq:cond1} is used to guarantee recursive feasibility of the optimization problem by ensuring that there exists an admissible controller that makes the terminal set positively invariant under the nominal dynamics. The second condition provides a sufficient condition to ensure the value function of the optimal control problem is a Lyapunov function, and thus guarantees on convergence of the nominal system to the origin can be obtained.

To ensure conditions \eqref{eq:cond1} and \eqref{eq:cond2} are satisfied, we design $\kappa(\bar{x})$, $P$, and $\mathcal{X}_N$ by considering the unconstrained infinite-horizon LQR problem with cost matrices $Q$ and $R$ for the nominal system dynamics \eqref{eq:nomdyn}. Specifically, we choose $P$ to be the solution to the associated discrete algebraic Riccati equation, and define $\kappa(\bar{x}) = K_f \bar{x}$ where $K_f$ is the associated LQR gain matrix. We then compute the set $\mathcal{X}_N$ which satisfies the conditions \eqref{eq:cond1}:
\begin{equation*}
A_{K_f}\mathcal{X}_N \subseteq \mathcal{X}_N, \quad H\mathcal{X}_N \subseteq \bar{\mathcal{Z}}, \quad K_f\mathcal{X}_N \subseteq \bar{\mathcal{U}},
\end{equation*}
where $A_{K_f} = A + BK_f$. Techniques for computing such as set are described in \cite{BorrelliBemporadEtAl2017} and \cite{GilbertTan1991}. Finally, by choice of $P$ and $K_f$ it can be seen that condition \eqref{eq:cond2} will hold with equality, and $P$ will be symmetric, positive definite. 

From the results in \cite{RawlingsMayneEtAl2017} we therefore can guarantee that the online optimization problem \eqref{eq:robustmpc} will be recursively feasible and the closed loop system defined by the nominal system dynamics \eqref{eq:nomdyn} under the control law $\bar{u}_k = \bar{u}^*_{k|k}$ will exponentially converge to the origin.

\subsection{Robust Constraint Satisfaction} \label{subsec:tightconst}
So far we have defined a control law \eqref{eq:controller} which seeks to drive the real system to track a nominal trajectory of the system \eqref{eq:nomdyn}, and we have defined an optimal control problem that exponentially drives the nominal system to the origin. We now discuss a technique for ensuring robust constraint satisfaction of the real system by computing bounds on the errors that could arise due to disturbances. These error bounds can then be used to tighten the constraint sets \eqref{eq:constraints} to give $\bar{\mathcal{Z}}$ and $\bar{\mathcal{U}}$.

Several types of errors are present in the system: $\hat{e}_k \coloneqq x_k - \hat{x}_k$ is the estimation error and $d_k \coloneqq \hat{x} - \bar{x}_k$ is the control error. The error between the real system state and the nominal state is defined as $e_k \coloneqq x_k - \bar{x}_k = \hat{e}_k + d_k$. Under the control law \eqref{eq:controller}, the coupled dynamics for these errors can be described by the system
\begin{equation} \label{eq:error}
\xi_{k+1}  = A_\xi \xi_k  + \delta_k, \\
\end{equation}
where
\begin{equation*}
A_\xi = \begin{bmatrix}
A-LC & 0 \\ LC & A+BK
\end{bmatrix}, \quad \delta_k = \begin{bmatrix}
I & -L \\ 0 & L 
\end{bmatrix}\begin{bmatrix}
w_k \\ v_k
\end{bmatrix},
\end{equation*}
and with $\xi_k = \begin{bmatrix} \hat{e}_k^T & d_k^T \end{bmatrix}^T$. Note that the matrix $A_\xi$ is Schur stable by the design of the gain matrices $K$ and $L$ and its block triangular structure. Additionally, by Assumption \ref{ass:convexcompact} the vector $\begin{bmatrix} w_k^T & v_k^T \end{bmatrix}^T$ is guaranteed to lie in a compact, convex set that contains the origin in its interior. Therefore it is straightforward to compute a convex, compact set $\Delta$ such that $\delta_k \in \Delta$.

Consider now a convex, compact RPI set $\mathcal{R} \coloneqq \{\xi \: | \: H_r \xi \leq b_r \}$ for the system \eqref{eq:error}. By definition we have the implication $\xi_0 \in \mathcal{R} \implies \xi_k \in \mathcal{R}$ for all $k>0$, which allows $\mathcal{R}$ to define the constant error ``tubes'' that will be used to tighten the constraint sets. Specifically, the tightened constraint sets are defined as
\begin{equation}
\bar{\mathcal{Z}} \coloneqq \mathcal{Z}\ominus \begin{bmatrix}
H & H
\end{bmatrix}\mathcal{R}, \quad \bar{\mathcal{U}} \coloneqq \mathcal{U} \ominus\begin{bmatrix}
0 & K
\end{bmatrix} \mathcal{R}.
\end{equation}

\subsection{Closed-loop System Properties}
We now state two important properties of the controlled system \eqref{eq:dyn} using our proposed scheme. The first property states that the system will satisfy the system constraints robustly, and the second is a result on convergence.
\begin{prop}[Robust Constraint Satisfaction]
Suppose $\xi_0 \in \mathcal{R}$ and that the optimal control problem \eqref{eq:robustmpc} is feasible at time $k=0$. Then, under all admissible disturbance sequences the system will satisfy the constraints $z_k \in \mathcal{Z}$ and $u_k \in \mathcal{U}$ for all $k\geq 0$.
\end{prop}
\begin{proof}
By design, the optimal control problem is recursively feasible and therefore the nominal system is guaranteed to satisfy $\bar{z}_k \in \bar{\mathcal{Z}}$ and $\bar{u}_k \in \bar{\mathcal{U}}$ for all $k\geq0$. Additionally, the assumption that $\xi_0 \in \mathcal{R}$ implies that $\xi_k \in \mathcal{R}$ for all $k>0$ since $\mathcal{R}$ is an RPI set. Finally, by definition of the tightened constraints and with $\xi_k \in \mathcal{R}$ it holds that $\bar{z}_k \in \bar{\mathcal{Z}} \implies z_k \in \mathcal{Z}$ and $\bar{u}_k \in \bar{\mathcal{U}} \implies u_k \in \mathcal{U}$.
\end{proof}

\begin{prop}[Convergence]
Suppose $\xi_0 \in \mathcal{R}$ and the optimal control problem \eqref{eq:robustmpc} is feasible at time $k=0$. Then the system converges exponentially to the set $[I\:\: I]\mathcal{R}$ under all admissible disturbance sequences.
\end{prop}
\begin{proof}
By design, the optimal control problem is recursively feasible and drives the nominal system \eqref{eq:nomdyn} to converge exponentially to the origin. Additionally, since $x_k = [I\:\: I]\xi_k + \bar{x}_k$ and since $\xi_0 \in \mathcal{R} \implies \xi_k \in \mathcal{R}$ for all $k\geq0$, it is true that $x_k \in \bar{x}_k \oplus [I\:\:I]\mathcal{R}$ for all $k\geq0$. Finally, since $\bar{x}_k\rightarrow 0$ exponentially we have that $x_k \rightarrow [I\:\:I]\mathcal{R}$ exponentially.
\end{proof}
Note that by choosing $\bar{x}_0 = \hat{x}_0$ (Section \ref{subsec:nomsys}) we have $d_0 = 0$ and therefore the assumption that $\xi_0 \in \mathcal{R}$ is only dependent on the estimator error $\hat{e}_0$. While $\hat{e}_0$ is not known in practice, it is standard to assume it is bounded at $k=0$.

\subsection{Computing $\mathcal{R}$}
It is desirable to compute an RPI set for \eqref{eq:error} that is as small as possible such that the constraint tightening is less conservative and so that the convergence guarantees can be stronger. Obviously this also implies that the computed RPI set should be constraint admissible, such that $E\mathcal{R} \subseteq \Phi$ where 
\begin{equation}
\phi \coloneqq E \xi, \quad E \coloneqq \begin{bmatrix}H & H \\ 0 & K \end{bmatrix},
\end{equation}
and the set $\Phi \coloneqq \{\phi \: | \: H_\phi \phi \leq b_\phi \}$ is defined with
\begin{equation}
H_\phi = \begin{bmatrix}H_z & 0 \\ 0 & H_u \end{bmatrix}, \quad b_\phi = \begin{bmatrix}b_z \\ b_u \end{bmatrix}.
\end{equation}
Note that $\Phi$ is convex, compact, and contains the origin in its interior by Assumption \ref{ass:convexcompact}, and we assume that the pair $(A_\xi, E)$ is observable. As was mentioned in Section \ref{subsec:tightconst} the set $\Delta$ is also convex and compact. For our proposed RPI set computation method in Section \ref{sec:computingrpi}, we also require $\Delta$ to contain the origin in its interior, which is not guaranteed for \eqref{eq:error}. However this can easily be fixed by enlarging $\Delta$ by an arbitrarily small amount as needed. Finally, we also make an assumption (see \cite{KolmanovskyGilbert1998, SchulzeDarupTeichrib2019}) that the minimal RPI set for \eqref{eq:error} satisfies $E\mathcal{R}_\infty \subseteq \Phi^\mathrm{o}$ where $\Phi^\mathrm{o}$ denotes the interior of $\Phi$.

Under the above assumptions, several approaches for computing RPI sets for the system \eqref{eq:error} exist. The approach in \cite{KolmanovskyGilbert1998} is efficient, but will likely result in a large $\mathcal{R}$ which is undesirable as it will lead to overly conservative constraint tightening. The method in \cite{RakovicKerriganEtAl2005} could be employed to obtain a small RPI set, but this would require computationally expensive Minkowski additions. One computationally viable option that could yield a small $\mathcal{R}$ is given by \cite{SchulzeDarupTeichrib2019}, however in this work we choose to use a novel method that is described in Section \ref{sec:computingrpi} which is a combination of the methods presented in \cite{SchulzeDarupTeichrib2019} and \cite{Trodden2016}.


\section{Proposed RPI Set Computation Method} \label{sec:computingrpi}
We now present a novel method for computing RPI sets that is both simple and efficient, and can be used in the synthesis of the control scheme described in Section \ref{sec:proposed} as well as other tube-based MPC schemes. Using the same mathematical notation as in Section \ref{sec:notation}, we consider autonomous systems of the form \eqref{eq:rpidynamics} where the following assumptions are made:
\begin{ass} \label{ass:Delta_Phi}
Both sets $\Delta$ and $\Phi$ are convex, compact, and contain the origin in their interior.
\end{ass}
\begin{ass} \label{ass:stable_obsv}
The matrix $A_e$ is Schur stable and the pair $(A_e, E)$ is observable.
\end{ass}
\begin{ass} \label{ass:mRPI}
The minimal RPI set $\mathcal{R}_\infty$ satisfies $E\mathcal{R}_\infty \subseteq \Phi^\mathrm{o}$ where $\Phi^\mathrm{o}$ denotes the interior of $\Phi$.
\end{ass}
These are standard assumptions that will ensure our approach exhibits the properties described in Section \ref{subsec:algprop}. For example, the disturbance set compactness and system stability assumptions are required for invariant sets to exist, and the remaining assumptions are required to guarantee that we can define an invariant set with a finite number of hyper-planes. For further discussion on these assumptions see \cite{KolmanovskyGilbert1998}. Now, to provide some insight into our technique we will briefly review two previously developed methods.

\subsection{Schulze Darup and Teichrib, \cite{SchulzeDarupTeichrib2019}}
This work combines the advantages of both \cite{KolmanovskyGilbert1998} and \cite{RakovicKerriganEtAl2005} to yield an algorithm that is more efficient than \cite{RakovicKerriganEtAl2005} and can generate RPI sets that are better approximations to the minimal RPI set than \cite{KolmanovskyGilbert1998} (which was designed to compute \textit{maximal} RPI sets).

This is accomplished by first using the techniques presented in \cite{RakovicKerriganEtAl2005} to compute, for the user defined $\epsilon \geq 0$, the value $k^*$ such that $(1+\epsilon) A_e^{k^*}\Delta \subseteq \epsilon \Delta$. Then a container set $\mathcal{C}(\epsilon) \coloneqq \{e \:|\: H_c e \leq b_c\}$ is defined where $H_c \coloneqq H_\phi E$ and $b_c \coloneqq (1+\epsilon)\sum_{j=0}^{k^*-1} S_\Delta\big(H_cEA_e^j \big)$. The RPI set $\mathcal{P}_\infty(\epsilon)$ is then computed as the largest RPI set contained in $\mathcal{C}(\epsilon)$ using the approach in \cite{KolmanovskyGilbert1998} which recursively defines $\mathcal{P}_k \coloneqq \{e \:|\: H_{p,k} e \leq b_{p,k}\}$ by
\begin{equation} \label{eq:dandtrecursion}
\begin{split}
H_{p,k} = \begin{bmatrix}
H_{p,k-1} \\ H_\phi E A_e^k
\end{bmatrix},  \quad b_{p,k} = \begin{bmatrix}
b_{p,k-1} \\ r_k
\end{bmatrix}, \\
r_k = r_{k-1} - S_\Delta\Big(H_\phi EA_e^{k-1}\Big),
\end{split}
\end{equation}
with $H_{p,0} = H_\phi E$ and $b_{p,0} = r_0 = b_c$. From the results in \cite{KolmanovskyGilbert1998} and \cite[Thm~1]{SchulzeDarupTeichrib2019} this recursion will terminate (i.e. $\mathcal{P}_{k+1} = \mathcal{P}_k$) in a \textit{finite} number of iterations $\bar{k}$ under the stated assumptions. The RPI set is then given as $\mathcal{P}_\infty(\epsilon) = \mathcal{P}_{\bar{k}}$.

Not only is this algorithm efficient, but from \cite[Thm~1]{SchulzeDarupTeichrib2019} it is proven that the choice of the container set $\mathcal{C}(\epsilon)$ yields RPI sets comparable to those in \cite{RakovicKerriganEtAl2005} in that the resulting tightened constraints would be identical for both methods. This is advantageous for robust MPC since \cite{RakovicKerriganEtAl2005} can yield RPI sets that are arbitrarily close to the minimal RPI set, which reduces conservativeness.

\subsection{Trodden, \cite{Trodden2016}}
Another approach, described in \cite{Trodden2016}, computes an RPI set defined as $\mathcal{T}(H_t) \coloneqq \{e \:|\: H_t e \leq b_t\}$, where the matrix $H_t$ is determined \textit{a priori} and the vector $b_t$ is determined by solving a linear program. They show \cite[Thm~4]{Trodden2016} that for a specific $H_t$, if an RPI set exists, then $\mathcal{T}(H_t)$ is the smallest RPI set with the chosen $H_t$. Thus, while the approach is simple and efficient, it requires careful consideration of the chosen $H_t$ such that an RPI set \textit{exists} and so that it is not too conservative. Unfortunately, no insightful guidelines for choosing $H_t$ are provided in \cite{Trodden2016}.

\subsection{Proposed Method}
In this work we propose to combine the mutually beneficial ideas from both \cite{Trodden2016} and \cite{SchulzeDarupTeichrib2019}. Specifically, we use insights from \cite{SchulzeDarupTeichrib2019} to identify a good set of hyper-planes that will define $H_r$, and use \cite{Trodden2016} to find the smallest RPI set associated with that choice. To accomplish this we propose Algorithm \ref{alg}
\begin{algorithm}
\caption{Compute $\mathcal{R}(k)$} \label{alg}
\begin{algorithmic}[1]
\Procedure{ComputeRPISet}{$k$}
\State $H_{r,0} \leftarrow H_\phi E$
\For{$i \in [1,\dots,k]$}
\State $H_{r,i} \leftarrow \begin{bmatrix}
H_{r,i-1} \\ H_\phi E A_e^k
\end{bmatrix}$
\EndFor
\State Solve linear program \eqref{eq:rpiLP} with $H_r = H_{r,k}$
\If{\eqref{eq:rpiLP} has unbounded objective}
\State \textbf{return} Failure
\Else
\State $b_{r,k} = c^* + d^*$
\State $\mathcal{R}(k) \leftarrow \{e \:|\: H_{r,k} e \leq b_{r,k} \}$
\State \textbf{return} $\mathcal{R}(k)$
\EndIf
\EndProcedure
\end{algorithmic}
\end{algorithm}
\vspace{-.10in}
where the linear program is defined by \eqref{eq:rpiLP}
\begin{equation} \label{eq:rpiLP}
\begin{split}
(c^*, d^*) = \underset{c, d, \xi_i, \omega_i}{\text{arg max.}} \:\:& \sum_{j=1}^{n_r}c_j + d_j, \\
\text{subject to} \:\: & c_i \leq h^T_{r,i} A_e \xi_i,  \\
& H_r\xi_i \leq c + d, \\
& d_i \leq h^T_{r,i} \omega_i, \\
& H_\delta \omega_i \leq b_\delta, \\
\end{split}
\end{equation}
where $i = 1, \dots, n_r$ and $h^T_{r,i}$ is the $i^\text{th}$ row of $H_r$.

As can be seen this algorithm is simple and efficient, as the RPI set $\mathcal{R}(k)$ only requires computation of a single linear program. As was previously mentioned, this method is mainly hindered by the assumption that a good choice for $H_r$ is known \textit{a priori} and that an RPI set exists for that choice. The definition of $H_r$ in Algorithm \ref{alg} along with insights from \cite{SchulzeDarupTeichrib2019} bridge this assumption. If for a chosen $k$, Algorithm \ref{alg} is not successful, then no RPI set exists for $H_r$ but the practitioner can simply increase the value of $k$ until a valid solution is found. We now discuss several useful properties of the approach.

\subsection{Algorithm Properties} \label{subsec:algprop}
The first important property of Algorithm \ref{alg} states that there exists a \textit{finite} $k$ such that the algorithm will return a valid RPI set.
\begin{thm} \label{thm:existence}
Suppose Assumptions \ref{ass:Delta_Phi}, \ref{ass:stable_obsv}, and \ref{ass:mRPI} hold. Then, there exists a finite integer $\underline{k}$ such that Algorithm \ref{alg} will return a valid RPI set for all $k\geq\underline{k}$.
\end{thm}
\begin{proof}
By \cite[Thm~4]{Trodden2016}, problem \eqref{eq:rpiLP} admits a bounded optimal solution if in addition to Assumptions \ref{ass:Delta_Phi} and \ref{ass:stable_obsv} it also holds for the chosen $H_r$ that: (i) an RPI set exists and (ii) the RPI set contains the origin in its interior. We first prove that an RPI set with $H_r = H_{r,k}$ exists for all $k \geq \underline{k}$ for some finite integer $\underline{k}$.

Using the results from \cite{SchulzeDarupTeichrib2019} (which leverage \cite[Thm~6.3]{KolmanovskyGilbert1998}) along with Assumptions \ref{ass:Delta_Phi}, \ref{ass:stable_obsv}, and \ref{ass:mRPI}, there is guaranteed to exist a \textit{finite} value $\bar{k}$ such that the set $\mathcal{P}_{\bar{k}}$ defined by \eqref{eq:dandtrecursion} is an RPI set and $\mathcal{P}_{k+1} = \mathcal{P}_{k}$ for all $k\geq \bar{k}$. Since the sets $\mathcal{P}_{k}$ defined in \eqref{eq:dandtrecursion} and $\mathcal{R}(k)$ defined by Algorithm \ref{alg} use the same hyper-planes (i.e. $H_{p,k} = H_{r,k}$) for all $k \geq 0$), it is apparent that for all $k\geq \bar{k}$ an RPI set exists for the choice of $H_{r,k}$. 

We now prove that the RPI sets with $H_r = H_{r,k}$ contain the origin in their interior for all $k\geq \bar{k}$. By the theorem assumptions the set $\Delta$ contains the origin in its interior and therefore by \cite[Thm~4.1]{KolmanovskyGilbert1998} the minimal RPI set $\mathcal{R}_\infty$ contains the origin in its interior as well. Thus, since $\mathcal{R}(k)$ is a valid RPI set for all $k\geq \bar{k}$ and since $\mathcal{R}_\infty \subseteq \mathcal{R}(k)$ by \cite[Cor~4.2]{KolmanovskyGilbert1998} we have the desired result.
\end{proof}

Next, we prove that the RPI sets generated by Algorithm \ref{alg} are non-increasing in size as $k$ increases.
\begin{thm} \label{thm:decrease}
Suppose Assumptions \ref{ass:Delta_Phi}, \ref{ass:stable_obsv}, and \ref{ass:mRPI} hold. Then, for all $k \geq \max \{\underline{k}, n_e-1 \}$, the sets $\mathcal{R}(k)$ and $\mathcal{R}(k+1)$ computed by Algorithm \ref{alg} will satisfy $\mathcal{R}(k+1) \subseteq \mathcal{R}(k)$.
\end{thm}
\begin{proof}
First, by Theorem \ref{thm:existence}, Algorithm \ref{alg} is guaranteed to return RPI sets $\mathcal{R}(k+1)$ and $\mathcal{R}(k)$ for all $k \geq \max \{\underline{k}, n_e-1 \}$, which also implies that the vectors $b_{r,k}$ are finite. We use this fact to first show that $\mathcal{R}(k)$ is a compact set.

The RPI set $\mathcal{R}(k)$ is defined by hyper-planes which are given as the rows of $H_{r,k}$. Additionally, the matrix $H_{r,k}$ can be written as
\begin{equation*}
H_{r,k} =  \pmb{H}_\phi O_k, \:\: \pmb{H}_\phi = \begin{bmatrix}
    H_\phi & & \\
    & \ddots & \\
    & & H_\phi
  \end{bmatrix}, \:\: O_k \coloneqq \begin{bmatrix}
E \\ \vdots \\ EA_e^k
\end{bmatrix}.
\end{equation*}
Since it is assumed that $\Phi$ is compact, a set given by $H_\phi \phi \leq b$ is also compact for any finite $b\geq0$, which then implies that the set $\{\pmb{\phi} \:|\: \pmb{H}_\phi \pmb{\phi} \leq b_{r,k} \}$ is also compact and thus $\pmb{\phi}$ is bounded. Now, by the theorem assumptions the pair $(A_e, E)$ is observable and $k \geq n_e-1$ such that the matrix $O_k$ has rank $n_e$, which then implies that $e$ must be bounded since $\pmb{\phi} = O_k e$. Therefore $\mathcal{R}(k)$ is compact, which is now used to prove the main result.

For the RPI set $\mathcal{R}(k+1)$ the hyper-planes are defined by $H_{r,k+1}$ which can be written as
\begin{equation*}
H_{r,k+1} = \begin{bmatrix}
H_{r,k} \\ H_\phi E A_e^{k+1}
\end{bmatrix},
\end{equation*}
where it is apparent that the hyper-planes from $\mathcal{R}(k)$ are also included. Now consider the candidate RPI set $\mathcal{R}^\dagger \coloneqq \{e \:|\: H_{r,k+1} e \leq b^\dagger \}$ where $b^\dagger$ is defined as
\begin{equation*}
b^\dagger \coloneqq \begin{bmatrix}
b_{r,k} \\ \sup_{e \in \mathcal{R}(k)} H_\phi E A_e^{k+1}e
\end{bmatrix}.
\end{equation*}
With this choice, the vector $b^\dagger$ is finite since $b_{r,k}$ is finite and $\mathcal{R}(k)$ is compact, and furthermore $\mathcal{R}^\dagger = \mathcal{R}(k)$. Thus $\mathcal{R}^\dagger$ is a valid RPI set for \eqref{eq:rpidynamics} and has the same hyper-planes as $\mathcal{R}(k+1)$. By construction of the linear program \eqref{eq:rpiLP} it is then guaranteed that $b_{r,k+1} \leq b^\dagger$, which implies that $\mathcal{R}(k+1) \subseteq \mathcal{R}^\dagger = \mathcal{R}(k)$.
\end{proof}

An additional useful insight into the performance of Algorithm \ref{alg} can also be made by noting the relationship between Algorithm \ref{alg} and the method in \cite{SchulzeDarupTeichrib2019}. Consider the computation of the set $\mathcal{P}_\infty(\epsilon)$ for some $\epsilon$, where $\bar{k}(\epsilon)$ iterations of the recursion \eqref{eq:dandtrecursion} were required. Then the RPI set $\mathcal{R}(\bar{k}(\epsilon))$ is guaranteed to exist and it is guaranteed that $\mathcal{R}(\bar{k}(\epsilon)) \subseteq \mathcal{P}_\infty(\epsilon)$ since the hyper-planes defining $\mathcal{P}_\infty(\epsilon)$ are a subset of the hyper-planes defining $\mathcal{R}(\bar{k}(\epsilon))$.
Based on the results in \cite[Thm~1]{SchulzeDarupTeichrib2019} we can therefore conclude that increasing $k$ in Algorithm \ref{alg} will also lead to RPI sets that are comparable to arbitrarily close approximations to the minimal RPI set, which is a desirable property to minimize the conservativeness of the robust MPC constraint tightening.

\subsection{Comparison of Approaches} \label{subsec:compare}
Now that we have presented Algorithm \ref{alg} and identified some of its useful properties we will present some comparative results using the same problem as described in Example 1 from \cite{SchulzeDarupTeichrib2019} by conducting the following experiment. First we define a value of $\epsilon$, implement the method in \cite{SchulzeDarupTeichrib2019} to compute the RPI set $\mathcal{P}_\infty(\epsilon)$, and save the value $\bar{k}(\epsilon)$ that corresponds to the number of iterations of the recursion \eqref{eq:dandtrecursion}. Next we use Algorithm \ref{alg} and the value of $\bar{k}(\epsilon)$ to compute the RPI set $\mathcal{R}(\bar{k}(\epsilon))$. Finally, for completeness we also compare against an RPI set computed using \cite{Trodden2016} where the hyper-planes $H_t$ are defined as the sides of an $r$-sided regular polygon (as is used in \cite[Section IV-A]{Trodden2016}). The value of $r$ is chosen to be the number of hyper-planes used to define $\mathcal{R}(\bar{k}(\epsilon))$ and so we denote this RPI set as $\mathcal{T}(\bar{k}(\epsilon))$.

This comparison was repeated for three different values of $\epsilon$ and the results are shown in Figure \ref{fig:comparerpi} and Table \ref{tab:comparerpi}. As is expected (and which is true for all $\epsilon$), we see in Figure \ref{fig:comparerpi} that the RPI set $\mathcal{R}(\bar{k}(\epsilon)) \subseteq \mathcal{P}_\infty(\epsilon)$, which demonstrates an advantage that Algorithm \ref{alg} has over the approach in \cite{SchulzeDarupTeichrib2019}. Interestingly we see that the set $\mathcal{T}(\bar{k}(\epsilon))$ also provides good results in this case. However it is important to note that the approach used to define the hyper-planes for $\mathcal{T}(\bar{k}(\epsilon))$ does not scale well with problem dimension and also there are no guarantees that an RPI set would be found. In Table \ref{tab:comparerpi} we also present results on how much the constraint in the directions of $x_1$, $x_2$, and $u$ would be tightened when compared against the set $\mathcal{R}(\bar{k}(\epsilon))$. In other words, a positive value would mean the tightened constraints are more conservative.
\begin{figure*}
    \centering
    \begin{subfigure}{.32\textwidth}
        \includegraphics[height=.8\textwidth]{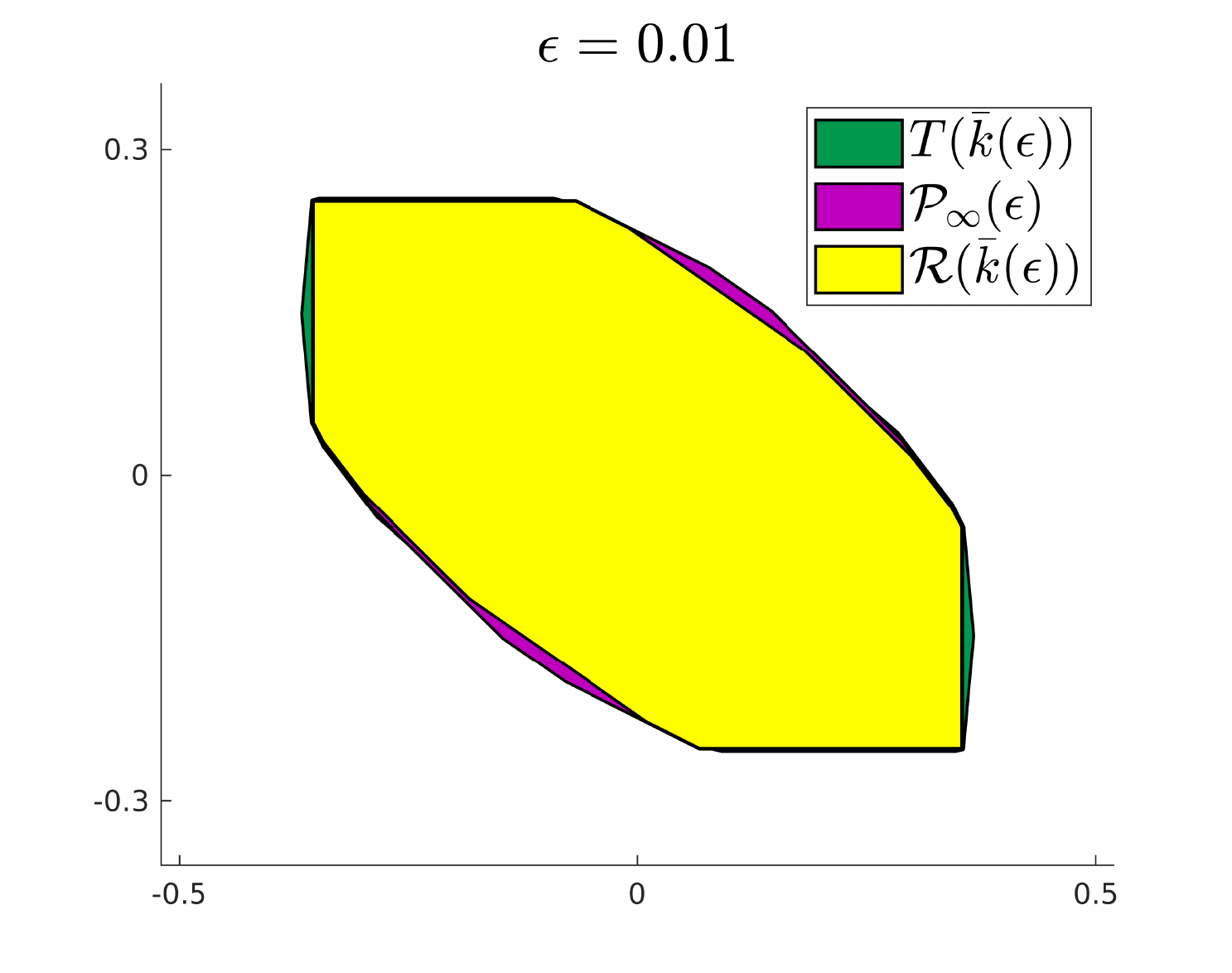}
    \end{subfigure}
    \begin{subfigure}{.32\textwidth}
        \includegraphics[height=.8\textwidth]{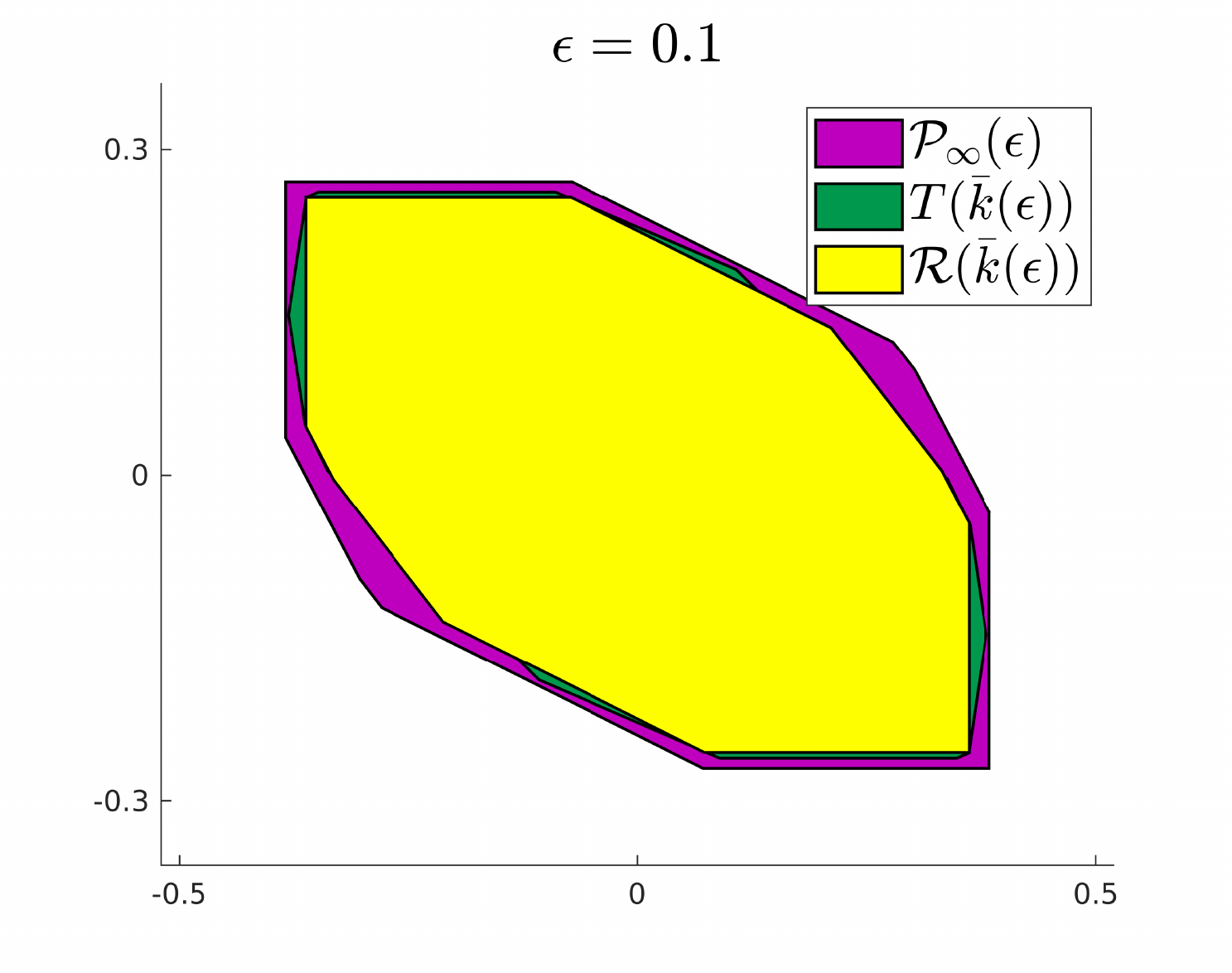}
    \end{subfigure}
    \begin{subfigure}{.32\textwidth}
        \includegraphics[height=.8\textwidth]{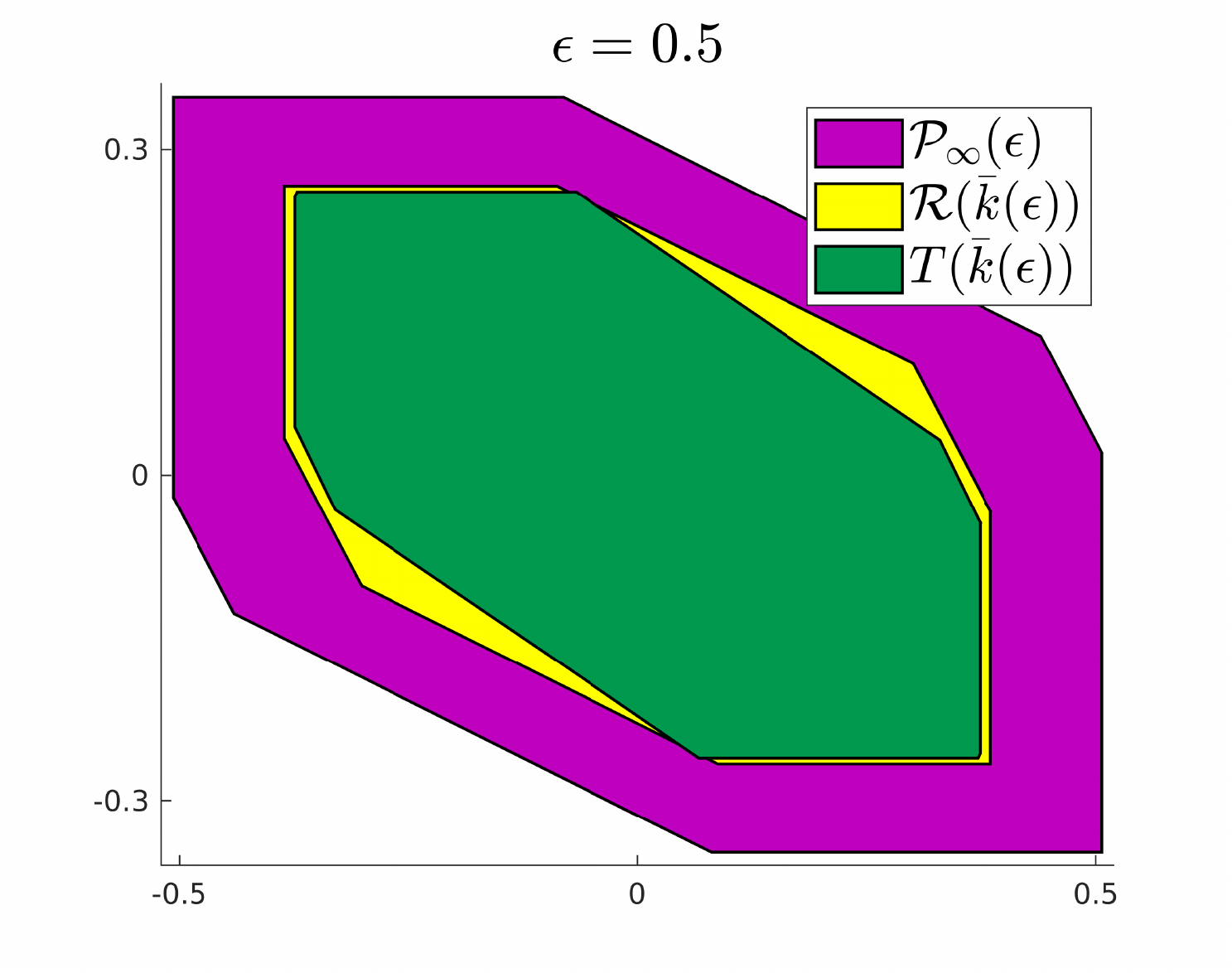}
    \end{subfigure}
    \caption{RPI sets computed for the comparison discussed in Section \ref{subsec:compare} for $\epsilon \in [0.01, 0.1, 0.5]$. The sets in yellow are computed using Algorithm \ref{alg}, the sets in purple are computed using the method in \cite{SchulzeDarupTeichrib2019}, and the sets in green are computed using \cite{Trodden2016} with hyper-planes defined by the sides of a regular polygon.}
\label{fig:comparerpi}
\vspace{-.20in}
\end{figure*}

\begin{table}[htbp]
\centering
\setlength{\tabcolsep}{4pt}
\begin{tabular}{l|c|c|c|c|c|c|c|c|c|} \cline{2-10} & \multicolumn{3}{c|}{$\epsilon = 0.01$}   & \multicolumn{3}{c|}{$\epsilon = 0.1$}   & \multicolumn{3}{c|}{$\epsilon = 0.5$}  \\ 
\cline{2-10}
\multicolumn{1}{l|}{}                                   & \multicolumn{1}{c|}{$\delta x_1$}   & \multicolumn{1}{c|}{$\delta x_2$}   & \multicolumn{1}{c|}{$\delta u$}      & $\delta x_1$   & $\delta x_2$   & $\delta u$      & $\delta x_1$   & $\delta x_2$   & $\delta u$  \\ \cline{1-10}
\multicolumn{1}{|r|}{{\color{mypurple} $\mathcal{P}_\infty(\epsilon)$}, $(\%)$} & \multicolumn{1}{c|}{{\color{mypurple} 0.3}} & \multicolumn{1}{c|}{{\color{mypurple} 0.3}} & \multicolumn{1}{c|}{{\color{mypurple} 0.3}} & {\color{mypurple} 6.1} & {\color{mypurple} 5.7} & {\color{mypurple} 6.7} & {\color{mypurple} 31.5} & {\color{mypurple} 30.6} & {\color{mypurple} 36.8} \\ \cline{1-10}
\multicolumn{1}{|r|}{{\color{mygreen} $\mathcal{T}(\bar{k}(\epsilon))$}, $(\%)$} & \multicolumn{1}{c|}{{\color{mygreen} 3.7}} & \multicolumn{1}{c|}{{\color{mygreen} 1.1}} & \multicolumn{1}{c|}{{\color{mygreen} 0.2}} & {\color{mygreen} 5.2} & {\color{mygreen} 2.0} & {\color{mygreen} 4.3} & {\color{mygreen} -3.0} & {\color{mygreen} -2.0} & {\color{mygreen} 1.5} \\ \cline{1-10}
\end{tabular}
\caption{Results on constraint tightening based on the RPI sets computed for the comparison discussed in Section \ref{subsec:compare}. The reported values represent the \textit{percentage} change in the amount a constraint is tightened when $\mathcal{P}_\infty(\epsilon)$ or $\mathcal{T}(\bar{k}(\epsilon))$ is used with respect to when $\mathcal{R}(\bar{k}(\epsilon))$ is used. A positive number means the resulting constraint is more conservative.}
\label{tab:comparerpi}
\vspace{-.20in}
\end{table}

\section{Examples} \label{sec:examples}
In this section we demonstrate the combined use of the robust MPC scheme developed in Section \ref{sec:proposed} and the proposed RPI set computation method described in Section \ref{sec:computingrpi}. We present results for two example systems: one is a synthetic system with state dimension $n=2$ and the second is a wind energy conversion system with $n=10$. In the implementation of both examples we leverage the open source Matlab toolboxes MPT3 \cite{HercegKvasnicaEtAl2013} and YALMIP \cite{Loefberg2004}, and solve all optimization problems using IBM ILOG CPLEX.

\subsection{Synthetic System} \label{subsec:synthetic}
For our first example we use the system and problem definition given in \cite{MayneRakovicEtAl2006}, and choose the matrix $K$ to be the LQR gain matrix with $Q = I$, $R = 0.01$, and $L$ is chosen to be the LQR gain matrix with $R = I$. Additionally, all simulated disturbances are computed uniformly at random within their defined bounds. We synthesize the robust MPC controller using our proposed RPI set computation method and choose a horizon of $N=13$. Some simulation results are shown in Figure \ref{fig:synthz}.
\begin{figure}[ht]
    \centering
    \includegraphics[width=0.8\columnwidth]{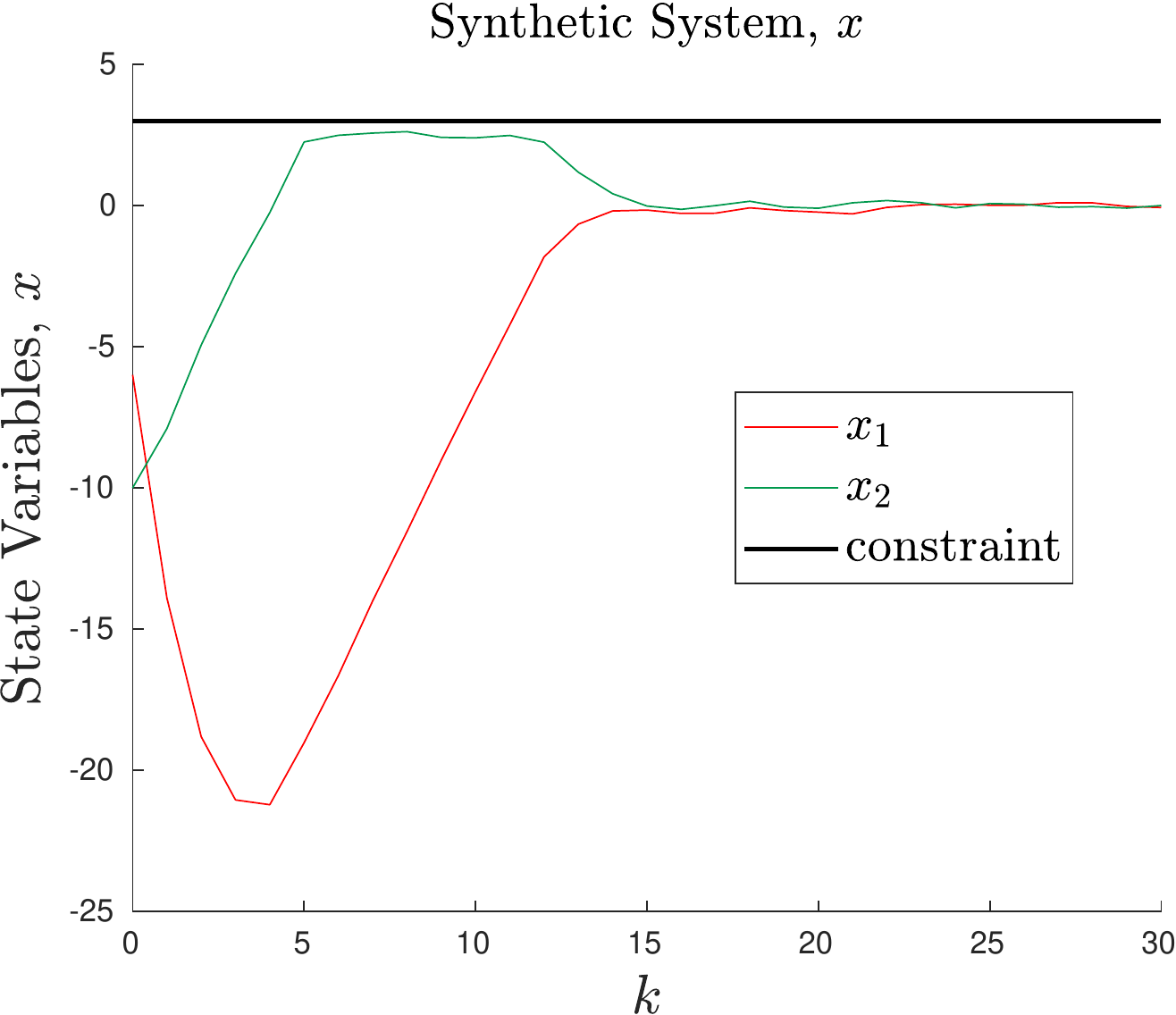}
    \caption{Simulation of the synthetic system described in Section \ref{subsec:synthetic}. It can be seen that the system satisfies the given constraints.}
    \label{fig:synthz}
\end{figure}

For comparison we also implement the control scheme in \cite{MayneRakovicEtAl2006}, and compute the required RPI sets using the approach in \cite{RakovicKerriganEtAl2005} with $\epsilon = 0.01$. First, in Table \ref{tab:synth_tight} we present a comparison of how much each constraint was tightened. The increase in tightening that is seen in \cite{MayneRakovicEtAl2006} can primarily be attributed to the fact that the RPI sets are computed \textit{sequentially}, instead of simultaneously as in our proposed approach.
\begin{table}
\centering
\begin{tabular}{l|c|c|c|}
\cline{2-4}
 & $x_1$ & $x_2$ & $u$ \\ \hline
\multicolumn{1}{|r|}{{\color{black} Mayne et al., $(\%)$}} & {\color{black} 33.2} & {\color{black} 81.8} & {\color{black} 59.3}  \\ \hline
\end{tabular}
\caption{Comparing the constraint tightening using our proposed approach and the robust MPC scheme from \cite{MayneRakovicEtAl2006}. The reported values are \textit{percentage} increases in how much each constraint is tightened with respect to our proposed control scheme.}
\label{tab:synth_tight}
\end{table}

\begin{figure}[ht]
    \centering
    \includegraphics[width=0.8\columnwidth]{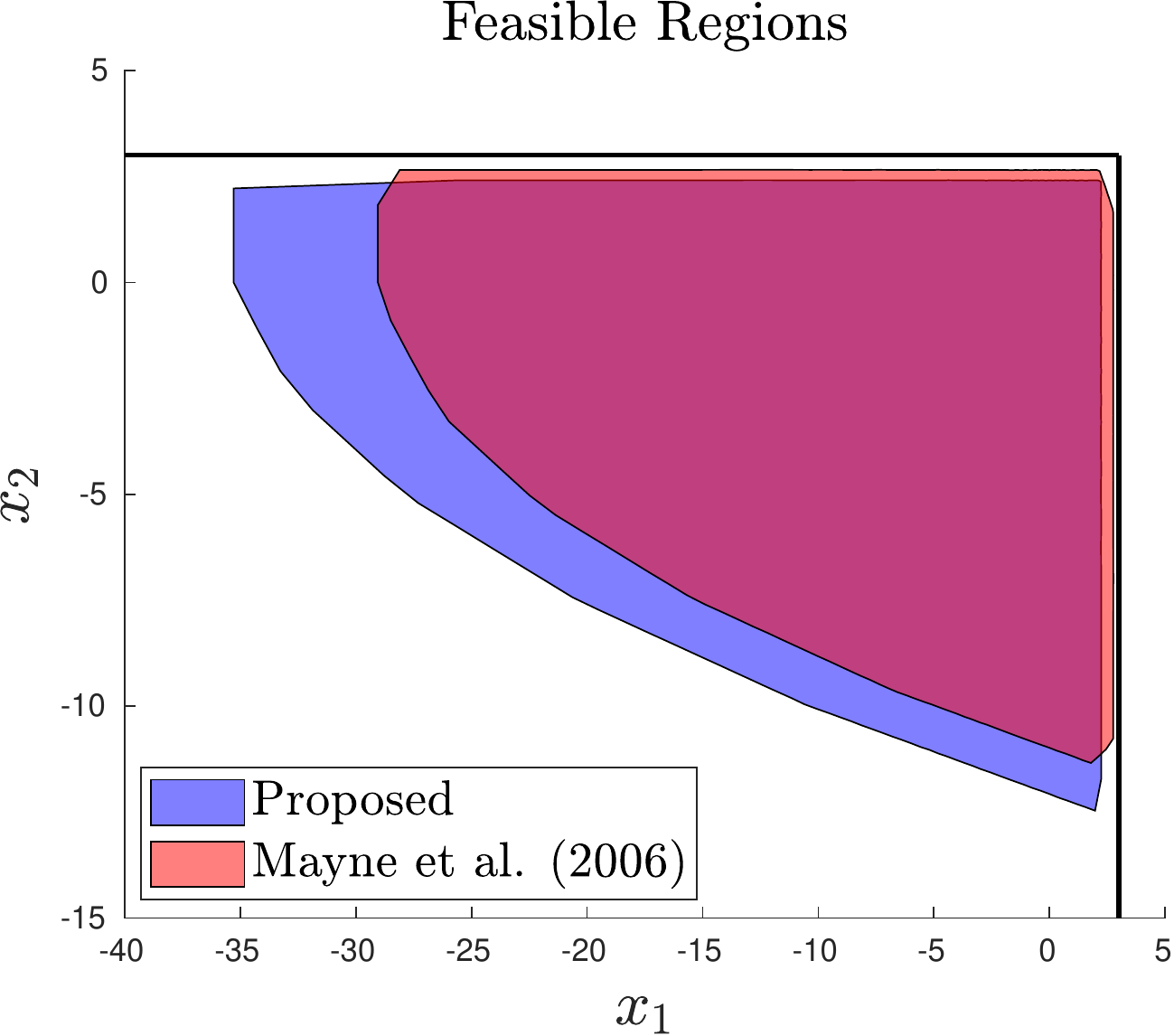}
    \caption{Feasible regions for the example discussed in Section \ref{subsec:synthetic}, using both our proposed control scheme and the scheme presented in \cite{MayneRakovicEtAl2006}.}
    \label{fig:feasible_region}
\end{figure}
Second, in Figure \ref{fig:feasible_region} we show a comparison of the approximate feasible regions for the two control schemes when both approaches were initialized with $x_0 = \hat{x}_0$. As can be seen the feasible region for \cite{MayneRakovicEtAl2006} is larger closer the constraint boundaries $x_1 \leq 3$ and $x_2 \leq 3$. This is mainly a result of the initial condition being a decision variable in the optimization problem, whereas in our approach we initialize $\bar{x}_0 = \hat{x}_0$. However the feasible region for our proposed approach is larger elsewhere, which is likely due to less constraint tightening leading to a larger terminal set $\mathcal{X}_N$.

Finally, we also compare the incurred cost of each approach by running simulations for $1000$ randomly sampled initial conditions that lie in the feasible regions for both methods. For each initial condition we set $x_0 = \hat{x}_0$ and run a disturbance free simulation with $k_f = 50$ time steps for both methods. The cost is computed using \eqref{eq:costfunc} over the finite horizon $k_f$. For this experiment the control scheme from \cite{MayneRakovicEtAl2006} incurred a cost that was on average $41\%$ higher than our proposed control scheme.

\subsection{Wind Energy Conversion System} \label{subsec:wec}
This system, described in \cite{Steinbuch1989}, includes models of the aerodynamics, rotor dynamics, drive train dynamics, and generator dynamics of a wind energy conversion system. The inputs to the system are a commanded rotor blade pitch angle, $\beta_r$, the commanded voltage of the generator output $u_{Fr}$, and the delay angle between the generator and grid voltages, $\alpha_r$. The system is linearized and has a state dimension of $n=10$. The measurements available include the output of a generator shaft speed sensor $\omega_{gm}$, the DC current from the generator $i_{dc}$, rotor speed $\omega_r$, and the mechanical torque in the generator shaft $T_m$. The control problem is to regulate the system to the nominal setpoint subject to bounded disturbances (with $\lVert w_k \rVert_\infty = \lVert v_k \rVert_\infty = .001$ for simplicity). Additionally we include constraints on the control, as well as on certain performance variables which include the relative angle of displacement in the generator shaft $\xi$, the generator shaft speed $\omega_g$, and the generator output current $i_{dc}$. The constraint on the relative angle of displacement in the generator shaft is used to control component fatigue, since the generator shaft is modeled as a flexible element. The model is discretized assuming a zero-order hold with sample time $\Delta t = 0.05$ seconds and the controller gains $K$ and $L$ are both given by the LQR gains computed with $Q=R=I$. The MPC cost function also uses these same weights. Simulated results for the performance variables $z$ can be seen in Figure \ref{fig:wecz}. 
\begin{figure}[ht]
    \centering
    \includegraphics[width=0.8\columnwidth]{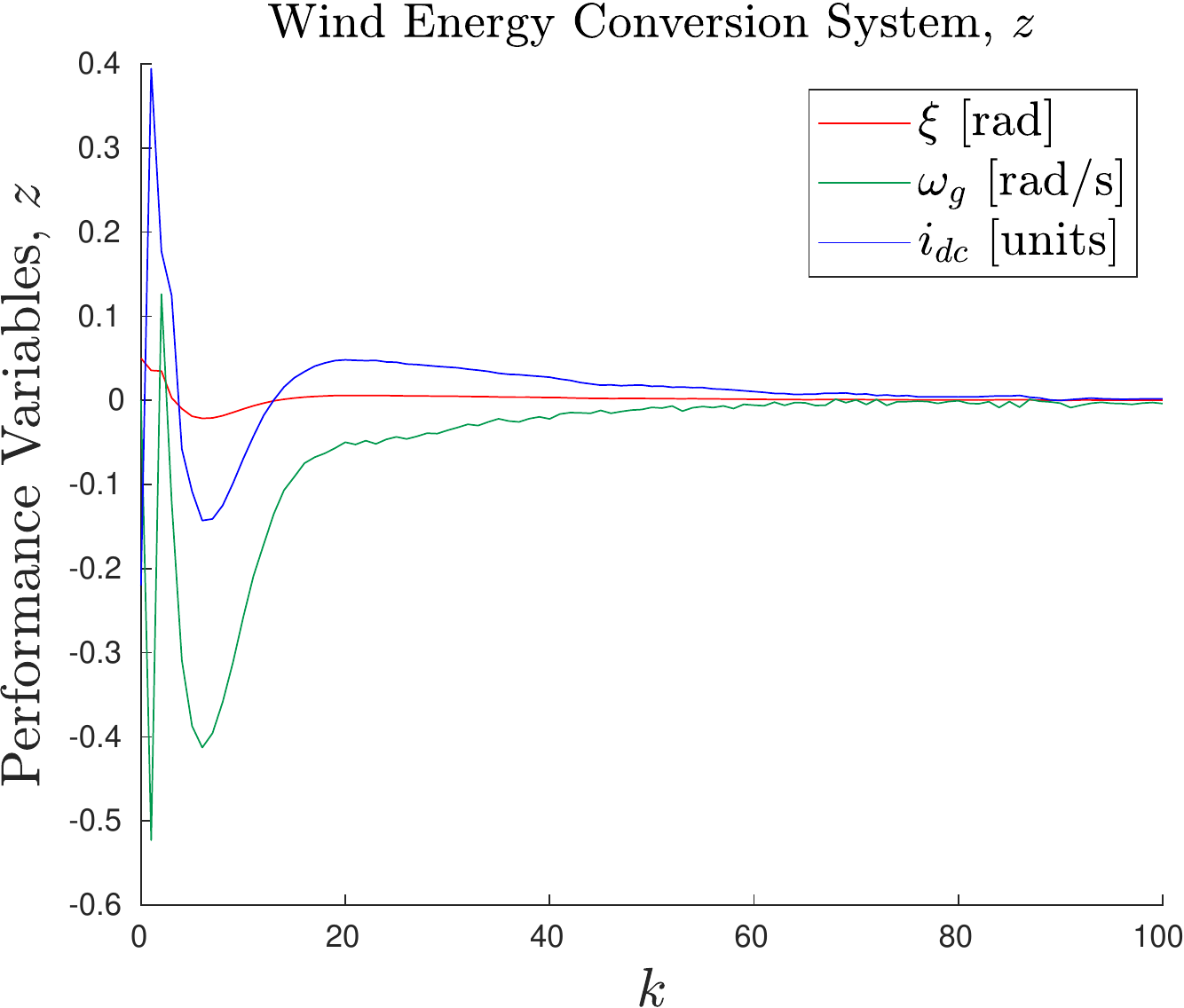}
    \caption{Simulated performance variable results for the wind energy conversion system described in Section \ref{subsec:wec}, including the relative angle of displacement in the generator shaft $\xi$, generator shaft speed $\omega_g$, and generator output current $i_{dc}$.}
    \label{fig:wecz}
\end{figure}

For this example system we can also compare our proposed RPI set computation method against that given in \cite{SchulzeDarupTeichrib2019}. Using the same comparison experiment discussed in Section \ref{subsec:compare} we obtain the results shown in Table \ref{tab:wecrpi}, which show that in such a comparison Algorithm \ref{alg} results in a less conservative RPI set than \cite{SchulzeDarupTeichrib2019} for the choice of $\epsilon = 0.01$.
\begin{table}[htbp]
\centering
\begin{tabular}{l|c|c|c|c|c|c|}
\cline{2-7}
 & $z_1$ & $z_2$ & $z_3$ & $u_1$ & $u_2$ & $u_3$ \\ \hline
\multicolumn{1}{|r|}{{\color{mypurple} $\mathcal{P}_\infty(\epsilon)$}, $(\%)$} & {\color{mypurple} 0.8} & {\color{mypurple} 0.3} & {\color{mypurple} 0.7} & {\color{mypurple} 0.9} & {\color{mypurple} 0.8} & {\color{mypurple} 0.8} \\ \hline
\end{tabular}
\caption{Comparing the constraint tightening using the RPI set from Algorithm \ref{alg} and the method in \cite{SchulzeDarupTeichrib2019} for the example in Section \ref{subsec:wec} with $\epsilon = 0.01$. The results are presented as \textit{percentage} increase in the amount of tightening compared to our proposed scheme. Thus a positive number indicates the RPI set yields more conservative tightened constraints.}
\label{tab:wecrpi}
\end{table}

\section{Conclusion} \label{sec:conclusion}
In this work we presented a tube-based robust output feedback MPC scheme that leads to efficient \textit{offline} controller synthesis and an efficient \textit{online} implementation. The efficiency of our approach was demonstrated to robustly control a wind energy conversion system with state space dimension $n=10$. In this work we also proposed a novel method for computing robust positively invariant sets which is simple, efficient, and is demonstrated to be effective when used for the proposed MPC scheme.

{\em Future Work:} Recent work in robust MPC has also yielded tube-based approaches where the tubes are not constant in time. It would be interesting to explore if a similar approach may be feasible for the output feedback setting based on our proposed RPI computation method. Additionally, it would be valuable to explore other properties of our proposed RPI computation method, such as \textit{a priori} bounds on required value of $k$ for algorithm success.

\bibliographystyle{IEEEtran}
\bibliography{main}

\end{document}